\pdfoutput=1
\documentclass[a4paper,11pt]{article}
\usepackage[T1]{fontenc}
\usepackage{lmodern}
\usepackage{amsmath,amssymb}
\usepackage{amsthm}
\usepackage{booktabs}
\usepackage{geometry}
\usepackage{graphicx}
\usepackage{float}
\usepackage[hidelinks]{hyperref}
\usepackage{bookmark}
\hypersetup{unicode=true}

\newtheoremstyle{uprighttheorem}
  {6pt}{6pt}{\normalfont}{}{\bfseries}{.}{0.5em}{}
\theoremstyle{uprighttheorem}
\newtheorem{definitioninner}{Definition}
\newtheorem{lemmainner}{Lemma}
\newtheorem{propositioninner}{Proposition}
\newtheorem{theoreminner}{Theorem}

\newenvironment{theorem}[1][]{\begin{theoreminner}[#1]}{\end{theoreminner}}

\title{Rate-$2/3$ Girth-$8$ $(3,18)$-Regular Quantum LDPC Codes\\
from Two-Branch Finite-Field Bases and CPM Lifts}
\author{Koki Okada and Kenta Kasai\\Institute of Science Tokyo\\
\texttt{okada.k.3154@m.isct.ac.jp}\\
\texttt{kenta@ict.eng.isct.ac.jp}}
\date{}

\begin{document}
\maketitle

\begin{abstract}
We construct a rate-$2/3$ quantum low-density parity-check (LDPC) code from a $(3,18)$-regular two-branch
finite-field base and a circulant-permutation-matrix (CPM) lift of degree $P=101$.
The resulting code is a Calderbank--Shor--Steane (CSS) code with parameters $[[34542,23032,18]]$.
Its distance is established by a computer-assisted proof: an explicit nontrivial logical operator gives the upper
bound, and a complete symmetry-reduced enumeration excludes every nonzero vector of weight below 18 in both
parity-check kernels.  We also prove that every affine-in-$t$ CPM lift of the same base satisfying CSS
orthogonality has distance at most 18 for every prime lift degree $P>19$.
The construction has row weight 18 and column weight 3, and the Tanner graphs of $H_X$ and $H_Z$ separately have girth 8.
Decoder experiments with log-likelihood-ratio (LLR) joint belief propagation (BP) and deterministic
post-processing show no failures in $10^8$ trials at $p=0.01$, and a finite-length frame error rate (FER)
sweep estimates the transition near $p=0.029$.
\end{abstract}

\section{Introduction}

Classical low-density parity-check (LDPC) codes were introduced as sparse parity-check codes in~\cite{gallager1962ldpc},
and Tanner graphs became a standard language for describing their sparse constraints in~\cite{tanner1981recursive}.
In the quantum setting, Calderbank--Shor--Steane (CSS) codes~\cite{calderbank_shor1996good,steane1996multiple}
give a direct way to build stabilizer codes from two binary parity-check matrices, but the two matrices must satisfy
an exact orthogonality condition. This condition makes quantum LDPC design more restrictive than the classical case:
sparsity, rate, short-cycle control, and commutation must be achieved simultaneously.

Early work adapted sparse graphs to stabilizer codes~\cite{mackay2004sparse}. Product constructions then provided
systematic ways to obtain families with positive rate and growing distance, including hypergraph product
codes~\cite{tillich_zemor2014hypergraph}, balanced product codes~\cite{breuckmann_eberhardt2021balanced}, lifted
product codes~\cite{panteleev_kalachev2022almost}, asymptotically good quantum LDPC codes, and quantum Tanner
codes~\cite{panteleev_kalachev2022asymptotically,leverrier_zemor2022tanner}. These results establish the
asymptotic viability of quantum LDPC codes, while finite-length design still requires explicit choices of degree
distribution, orthogonality constraints, lift structure, and decoder.

High rate is also a practical requirement in fault-tolerant quantum computation (FTQC). The number of physical
qubits per protected logical qubit affects the size of memory blocks, the cost of logical operations, and the
hardware scale required for large computations. Constant-overhead FTQC can be formulated in terms of quantum LDPC
families~\cite{gottesman2014constant_overhead}, and quantum LDPC codes have been proposed as alternatives to
surface-code-style layouts for reducing overhead~\cite{breuckmann_eberhardt2021_qldpc_review}. Recent work on
reconfigurable atom arrays also emphasizes finite-size high-rate codes with practical logical error rates
~\cite{zhao_2026_ultra_high_rate}. For this reason, an explicit high-rate finite-length code is useful only if its
algebraic constraints and decoder performance are both checked directly.

Finite-length performance depends on decoding as well as construction. Belief propagation combined with
ordered-statistics-decoding post-processing has been used successfully for degenerate quantum LDPC
codes~\cite{panteleev_kalachev2021degenerate,roffe2020decoding}, building on ordered statistics decoding for
classical linear codes~\cite{fossorier_lin1995osd}. Trapping-set analysis for quantum LDPC
codes~\cite{raveendran_vasic2021trapping} gives another way to analyze small residual syndromes after iterative
decoding. Recent finite-length work has also studied non-binary LDPC-based quantum error correction near the
coding-theoretic bound~\cite{komoto_kasai_2025_near_bound}, error-floor reduction, and rapid finite-length
error-rate transitions under joint belief-propagation decoding
~\cite{kasai_2025_error_floors,komoto_kasai_2025_sharp_transitions}.

On the construction side, recent work has focused on the interaction between CSS orthogonality, regular sparse
matrices, short-cycle control, and finite circulant-permutation-matrix (CPM) or affine-permutation lifts
~\cite{kasai_2026_orthogonality_barrier,okada_kasai_2026_affine_coset,okada_kasai_2026_square_base}. The
two-branch finite-field base construction in~\cite{okada_kasai_2026_twobranch} separates the base-stage regularity
and orthogonality checks from the lift-stage search. This paper applies that approach to a column-weight-three,
row-weight-eighteen target with rate close to two thirds, and then gives a CPM lift of degree one hundred one.
The result is not presented as an asymptotic family. The contribution is the explicit construction of one high-rate
CSS LDPC instance, together with direct verification of its ranks, orthogonality, row and column weights, same-type
girth, exclusion of specified low-weight logical supports, decoding measurements, and the exact distance
$d=18$.  A structural argument further shows that 18 is an unavoidable upper bound for every affine-in-$t$ CPM
lift of this particular base satisfying CSS orthogonality when the prime lift degree exceeds 19.

\section{Two-Branch Finite-Field Base}
\label{sec:finite-field-base}

Throughout the paper, binary vectors are column vectors and are denoted by bold symbols when they are written as
vectors. The notation ``same-type Tanner graph'' means the Tanner graph of one parity-check matrix alone, either
$H_X$ or $H_Z$; this is a descriptive term used to distinguish those graphs from the bipartite incidence relation
between $X$- and $Z$-checks used in the CSS orthogonality test.

This section specifies the base matrices before the CPM lift. The role of the base is to fix the row and column
weights, to make CSS orthogonality a finite-field coset condition, and to remove same-type 4-cycles before any lift
coefficients are chosen.

\subsection{Field and Indices}

Let the field be $\mathbb{F}_{19}$, and let
\begin{equation*}
  M=\{1,4,16,7,9,17,11,6,5\}
\end{equation*}
be the order-9 multiplicative subgroup of $\mathbb{F}_{19}^{\times}$.
The target column weight is $J=3$, and the target row weight is $L=18$.

Base columns are indexed by
\begin{equation*}
  (b,t,h)\in \{0,1\}\times \mathbb{F}_{19}\times M,
\end{equation*}
where $b$ is the branch, $t$ is the field variable, and $h$ is the subgroup element.
Hence the base length is
\begin{equation*}
  n_0=2\cdot 19\cdot 9=342.
\end{equation*}
Rows on each side $S\in\{X,Z\}$ are indexed by
\begin{equation*}
  (g,r)\in \{0,1,2\}\times \mathbb{F}_{19},
\end{equation*}
so each side has $3\cdot 19=57$ base rows.

\subsection{Base Incidence Rule}

The coefficient matrices used in this construction are
\begin{equation*}
A=
\begin{pmatrix}
0&16&17\\
0&2&14
\end{pmatrix},
\qquad
B=
\begin{pmatrix}
4&10&11\\
11&10&5
\end{pmatrix}.
\end{equation*}
Here $A_{b,g}$ is used on the $X$ side and $B_{b,g}$ on the $Z$ side.

A column $(b,t,h)$ is connected on the $X$ side to the row $(g,r)$ satisfying
\begin{equation*}
  r=t+A_{b,g}h \pmod {19}.
\end{equation*}
On the $Z$ side, the same column is connected to the row $(g,r)$ satisfying
\begin{equation*}
  r=t+B_{b,g}h \pmod {19}.
\end{equation*}
For a fixed column, there is one edge for each $g=0,1,2$, so the column weight is 3.
For a fixed row, each choice of $b\in\{0,1\}$ and $h\in M$ determines a unique $t$, so the row weight is
$2|M|=18$.
The block structure of these base parity-check matrices is shown in Fig.~\ref{fig:base-code-bitmap}.

\begin{figure}[H]
\centering
\includegraphics[width=0.95\linewidth]{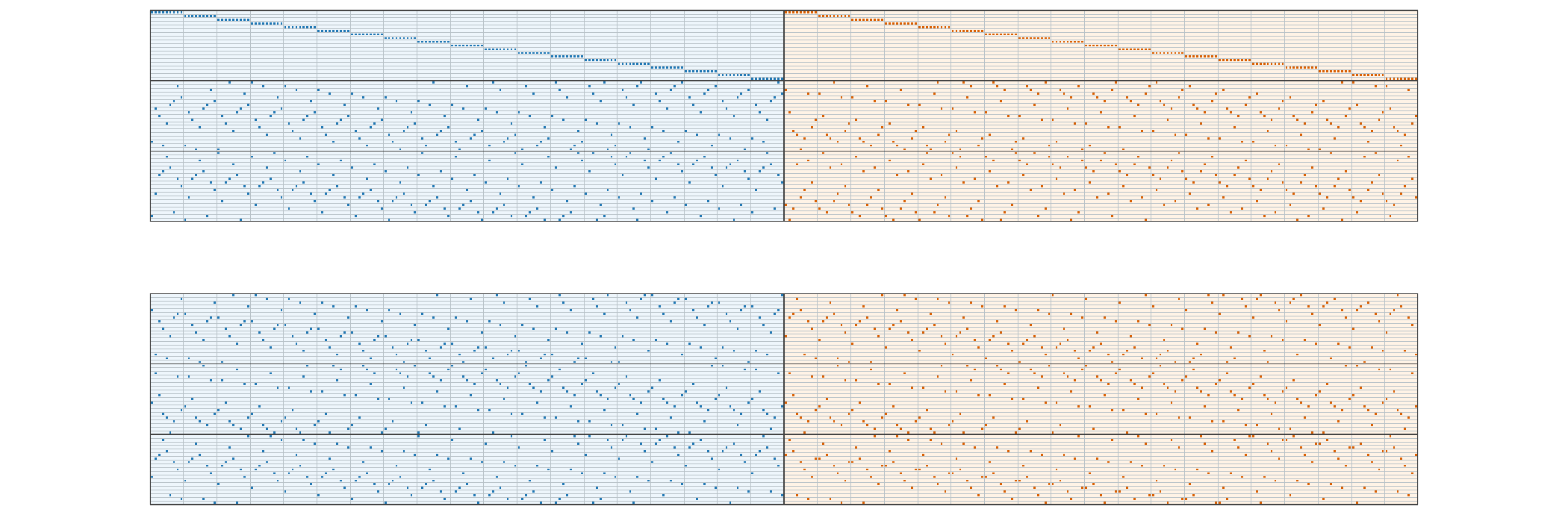}
\caption{Block-colored bitmap of the base parity-check matrices $H_X^{(0)}$ and $H_Z^{(0)}$.
The natural macroblock structure is $3$ row groups by $2$ column-branch blocks.  Each row group
consists of 19 finite-field-coordinate rows, and each column branch is subdivided by the finite-field
coordinate into 19 blocks of 9 columns.  The base CSS code has parameters $[[342,232,6]]$.}
\label{fig:base-code-bitmap}
\end{figure}

\subsection{Quotient-Coset Conditions}

Same-type 4-cycle exclusion and CSS orthogonality are checked by quotient-coset tests.
For the $X$ side and two distinct row groups $g\ne g'$, define
\begin{equation*}
  \Delta_0=A_{0,g}-A_{0,g'},\qquad
  \Delta_1=A_{1,g}-A_{1,g'}.
\end{equation*}
The same-type condition requires $\Delta_0,\Delta_1\ne0$ and
\begin{equation*}
  \Delta_0/\Delta_1\notin M.
\end{equation*}
The same condition is imposed on $B$ for the $Z$ side.
This ensures that any two rows on the same side intersect in at most one column, eliminating same-type 4-cycles.

For CSS orthogonality, for an $X$-row group $g$ and a $Z$-row group $g'$, define
\begin{equation*}
  \Delta_0=A_{0,g}-B_{0,g'},\qquad
  \Delta_1=A_{1,g}-B_{1,g'}.
\end{equation*}
The cross condition requires $\Delta_0,\Delta_1\ne0$ and
\begin{equation*}
  \Delta_0/\Delta_1\in M.
\end{equation*}
Then any $X$-row and $Z$-row intersect in either 0 or 2 columns, which is even over $\mathbb{F}_2$.
Therefore $H_XH_Z^T=0$ holds at the base level.
This coefficient pair passes these tests. Thus the base matrices already have the required row and column
weights, no same-type 4-cycles, and exact CSS orthogonality.

\section{CPM Lift and Construction Procedure}
\label{sec:cpm-lift}

The lift replaces each nonzero base entry by a circulant permutation matrix. The purpose of the lift is to increase
the blocklength while preserving the verified base-level constraints. The additional restriction introduced here is
that each shift is affine-linear in the finite-field coordinate $t$; this keeps the lift constraints linear over
$\mathbb{Z}_{101}$.

\subsection{Lift Definition}

Each 1-entry in the base matrix is replaced by a $P\times P$ cyclic permutation matrix, with $P=101$.
The lift degree is not fixed uniquely by the theory; it is a design parameter in the finite search used here.
We choose $P$ to be prime and larger than $q=19$ so that the lifted constraints can be handled by linear algebra over
the field $\mathbb{Z}_P$.
The value $P=101$ gives a code of length $342P$ while still allowing the search to find coefficients satisfying CSS
orthogonality, all same-type 6-cycle nonclosure constraints, and the exclusion constraints for low-weight logical
supports identified during the iterative measurements.
Increasing $P$ gives more possible values for the lift coefficients, but it also increases the blocklength and the
cost of verification and decoder measurements.
Thus $P=101$ is not claimed to be theoretically optimal for distance or threshold; it is the lift degree for which this
finite-length construction balances the construction constraints with computational tractability.
If the edge between base row $i$ and base column $j$ has shift $\sigma(i,j)\in\mathbb{Z}_{101}$, then lifted row
$(i,s)$ is connected to
\begin{equation*}
  (j,\,s+\sigma(i,j)\bmod P).
\end{equation*}

Rather than choosing all edge shifts independently, we restrict them to be affine-linear in the base variable $t$:
\begin{equation}
  \sigma_S(g,b,h;t)=c_{S,g,b,h}+d_{S,g,b,h}t \pmod {101}.
\label{eq:affine-shift}
\end{equation}
The number of variables is $2\cdot3\cdot2\cdot9\cdot2=216$.

The purpose of this affine-shift restriction is not only to reduce the number of lift variables.
It preserves the algebraic structure of the two-branch finite-field base, and it turns the main lifted constraints into
linear algebra over $\mathbb{Z}_{101}$.
In particular, lifted CSS orthogonality becomes a system of linear equations in the coefficients
$c_{S,g,b,h},d_{S,g,b,h}$, and each lifted same-type 6-cycle closure condition becomes the vanishing of a
linear form.
Thus the construction can first solve the orthogonality equations and then search the resulting affine solution space
while avoiding the nonzero linear forms corresponding to unwanted 6-cycles and identified low-weight logical-support
modes.
Choosing independent arbitrary CPM shifts would give more freedom, but it would also lose this finite-field
parametrization and make the simultaneous enforcement of orthogonality, girth, and forbidden-support constraints a
larger nonlinear search problem.

\subsection{Lifted Orthogonality}

Suppose a base $X$-row $x$ and a base $Z$-row $z$ intersect in two columns $j_0,j_1$.
To preserve even overlap after lifting, the two intersections must occur at the same sheet difference.
This gives the linear equation
\begin{equation}
  \sigma_X(x,j_0)-\sigma_Z(z,j_0)
  -\sigma_X(x,j_1)+\sigma_Z(z,j_1)=0\pmod {101}.
\label{eq:lifted-orthogonality}
\end{equation}
Substituting the affine-shift form in~\eqref{eq:affine-shift} into~\eqref{eq:lifted-orthogonality} gives a linear
equation in the lift coefficients. All such cross-row pairs produce 1539 linear equations.
After reduction, the equality system used for the lift search has rank 167, leaving 49 free variables.

\subsection{Excluding Same-Type 6-Cycles}

For each base same-type 6-cycle $C$, let $\ell_C(c,d)\in\mathbb{Z}_{101}$ be the signed sum of edge shifts
around the cycle.
The lifted cycle closes if and only if
\begin{equation}
  \ell_C(c,d)=0\pmod {101}.
\label{eq:six-cycle-closure}
\end{equation}
Thus all lifted same-type 6-cycles are excluded by requiring $\ell_C(c,d)\ne0$ for every base same-type 6-cycle in
\eqref{eq:six-cycle-closure}.

The base has 5472 same-type 6-cycles on each side, 10944 in total.
After reduction by the lifted-orthogonality equality system, these give 3749 distinct nonzero linear-form
constraints.
A randomized search over the 49-dimensional affine solution space found a point satisfying all of them.

\subsection{Eliminating Low-Weight Logical Supports Observed During Construction}
\label{sec:forbidden-supports}

The forbidden-support list is the set of base-column supports that are explicitly forbidden to have a consistent
realization in the next CPM lift.
It was obtained through an iterative construction procedure.
We first built a CPM lift and measured its finite-length decoding behavior.
When a low-weight logical failure was observed, its logical support was projected to base columns, constraints were
added so that this support and its affine equivalents would be incompatible with the next lift, and the lift search
was rerun.
The observed low-weight logical supports were not arbitrary small sets; they were produced by small even-intersection
patterns in the two-branch base.
Of the 110 input supports collected for this construction, 109 had lifted weight 6 and still projected to six base
columns.
Almost all of these weight-6 supports chose three columns from each of the two branches.
Such supports often meet opposite-side checks in 0 or 2 incident edges.
When this happens, the incident edges at each opposite-side check can be paired, and realizability as a lifted logical
support is governed mainly by linear consistency equations for the CPM sheet indices.
Thus the relevant structure is not a same-type short cycle itself, but a small base support with even opposite-side
check incidences whose sheet-consistency equations are satisfied by the chosen CPM shifts.
Because the finite-field coordinates have affine symmetries, each support identified by the measurements represents an orbit of equivalent
supports, so the affine equivalents were excluded together with it.
The supports accumulated through this iterative design procedure were organized as entries containing the side, an
occurrence count, the residual syndrome weight, the logical weight, and the support columns.
The support fields contain lifted column indices for putative $X$- and $Z$-type logical supports.
For lift degree $P=101$, a lifted column index $u$ is written as
\begin{equation*}
  u=Pj+s,\qquad 0\le s<P,
\end{equation*}
and projected to the base column $\pi(u)=j=\lfloor u/P\rfloor$ by discarding the sheet index $s$.
The base column $j$ is decoded as
\begin{equation*}
  j=b\,|\mathbb{F}_{19}|\,|M|+t\,|M|+i_h,\qquad
  b\in\{0,1\},\quad t\in\mathbb{F}_{19},\quad h_i\in M.
\end{equation*}
The 110 input supports reduce after projection and deduplication to 69 base supports
(39 on the $X$ side and 30 on the $Z$ side).
We also included the affine orbit induced by the finite-field symmetry
\begin{equation*}
  t\mapsto at+b,\qquad h\mapsto ah,\qquad a\in M,\ b\in\mathbb{F}_{19}.
\end{equation*}
After orbit expansion and deduplication, this produced 285 forbidden supports:
228 on the $X$ side and 57 on the $Z$ side.
The data specifying this forbidden-support exclusion step, together with the complete CPM lift coefficients of the
constructed code, are provided on the homepage cited in~\cite{kasai_homepage}.

For each forbidden base support, we tested directly whether it can be realized in the lifted code as a codeword or
logical-support pattern.
If the support meets each opposite-side check in an even number of edges, the incident edges at that check can be
paired.
Each choice of such pairings gives linear consistency equations for the CPM sheet indices.
If those equations are all satisfied by the chosen lift coefficients, then that base support still has a compatible
lifted realization and has not been eliminated.
The lift coefficients were chosen so that no such realization remains.
The verification found no forbidden base support for which any pairing pattern satisfied the corresponding
CPM-sheet consistency equations.
Thus none of the 285 forbidden base supports, including their affine equivalents, can be realized in the constructed
code.

\subsection{Construction Algorithm}
\label{sec:construction-algorithm}

The construction procedure combines the algebraic certificates in
Section~\ref{sec:finite-field-base} and the preceding parts of Section~\ref{sec:cpm-lift} into a single
finite-length CSS code.
The purpose is to preserve the base-level regularity and CSS orthogonality, remove lifted same-type 6-cycles, and also
avoid the low-weight logical-support patterns found during the iterative decoding measurements.
The main difficulty is that these requirements interact: CSS orthogonality is an equality constraint, same-type
6-cycle exclusion is a collection of nonzero linear-form constraints, and forbidden-support exclusion depends on
pairing choices inside the support.
The affine CPM shift form in~\eqref{eq:affine-shift} is used so that the first two lifted constraints can be handled by
linear algebra before the remaining nonzero and support-realizability tests are applied.
With this interpretation, the construction proceeds as follows.
\begin{enumerate}
  \item Build the finite-field base using the incidence rules $r=t+A_{b,g}h$ and $r=t+B_{b,g}h$ from
        Section~\ref{sec:finite-field-base}.
  \item Verify the quotient-coset certificate from Section~\ref{sec:finite-field-base}; this checks same-type
        4-cycle exclusion and CSS orthogonality before lifting.
  \item Introduce CPM lift variables $c_{S,g,b,h},d_{S,g,b,h}$ for the affine shift form in~\eqref{eq:affine-shift}.
  \item Enumerate all lifted CSS-orthogonality equations, using the condition in~\eqref{eq:lifted-orthogonality}.
  \item Enumerate the same-type 6-cycle closure conditions in~\eqref{eq:six-cycle-closure} and reduce them modulo the
        lifted-orthogonality equality system.
  \item Search the resulting 49-dimensional equality-system solution space for a point satisfying all nonzero
        6-cycle constraints.
  \item Generate lift-consistency equations for the forbidden support list and reject lift coefficients for which any
        forbidden support remains realizable, as described in Subsection~\ref{sec:forbidden-supports}.
  \item Expand $H_X,H_Z$ and explicitly verify ranks, orthogonality, girth, row weights, and column weights; the
        resulting code specification and decoder measurements are reported in Subsection~\ref{sec:resulting-code}
        and Section~\ref{sec:verification-measurements}.
\end{enumerate}

\subsection{Resulting Lifted Code}
\label{sec:resulting-code}

Expanding the base matrices with the selected CPM lift coefficients gives two binary parity-check matrices
$H_X$ and $H_Z$.
The base length is $n_0=342$ and the lift degree is $P=101$, so the lifted blocklength is
$n=Pn_0=34542$.
Each side has $57P=5757$ check rows, and the base incidence rule preserves column weight 3 and row weight 18.
Rank computation gives $\operatorname{rank}H_X=\operatorname{rank}H_Z=5755$, hence the CSS dimension is
$k=n-\operatorname{rank}H_X-\operatorname{rank}H_Z=23032$.
Thus the constructed lifted code is a near-rate-$2/3$ $(3,18)$-regular CSS LDPC code with parameters
$[[34542,23032,18]]$.
The main specifications are summarized in Table~\ref{tab:main-parameters}.

\begin{table}[H]
\centering
\caption{Main parameters of the constructed code.}
\label{tab:main-parameters}
\begin{tabular}{lr}
\toprule
Quantity & Value\\
\midrule
Base length $n_0$ & 342\\
Base rows per side & 57\\
Base $\operatorname{rank}H_X^{(0)}$ & 55\\
Base $\operatorname{rank}H_Z^{(0)}$ & 55\\
Base dimension $k_0$ & 232\\
Base distance $d_0$ & 6\\
Base CSS parameters $[[n_0,k_0,d_0]]$ & $[[342,232,6]]$\\
Lift degree $P$ & 101\\
Code length $n$ & 34542\\
Rows of $H_X$ & 5757\\
Rows of $H_Z$ & 5757\\
$\operatorname{rank}H_X$ & 5755\\
$\operatorname{rank}H_Z$ & 5755\\
$k=n-\operatorname{rank}H_X-\operatorname{rank}H_Z$ & 23032\\
Rate $k/n$ & 0.6667824677\\
Row weight & 18\\
Column weight & 3\\
Same-type girth & 8\\
Minimum distance $d$ & 18\\
\bottomrule
\end{tabular}
\end{table}

The base ranks are $\operatorname{rank}H_X^{(0)}=\operatorname{rank}H_Z^{(0)}=55$, giving base dimension
$k_0=232$ and base rate $0.6783625731$.
An exhaustive small-weight search finds no nontrivial base logical operator of weight below 6 and finds
weight-6 logical operators on both CSS sides; hence the base CSS parameters are
$[[n_0,k_0,d_0]]=[[342,232,6]]$.
The exact lifted distance and the structural limitation of the affine lift ansatz are established in
Subsection~\ref{sec:distance-upper-bound}.

\subsection{Exact Minimum Distance and an Affine-Lift Ceiling}
\label{sec:distance-upper-bound}

The affine dependence in~\eqref{eq:affine-shift} imposes a distance ceiling that is independent of the numerical
choice of its coefficients.

\begin{theorem}[Affine-lift distance ceiling]
For the two-branch base specified by $A$, $B$, and $M$ above, let $P>19$ be prime.  Every CPM lift whose shifts have
the affine form~\eqref{eq:affine-shift} and satisfy the lifted CSS-orthogonality equations
in~\eqref{eq:lifted-orthogonality} has quantum distance at most 18.
\end{theorem}

\begin{proof}
All arithmetic in the base coordinates is modulo 19.  Consider the 18 base columns
\begin{equation}
 \mathcal{S}=\{(0,6-6h,h),(1,6-9h,h):h\in M\}.
 \label{eq:unavoidable-support}
\end{equation}
Each of the 27 $X$-checks met by $\mathcal{S}$ contains exactly two of these columns.  Pairing the two columns at
each occupied check gives a connected graph with 18 vertices and 27 edges.  A choice of one lift index for every
column of $\mathcal{S}$ has zero $H_X$ syndrome precisely when the corresponding edge-difference equations are
consistent.  The graph has cycle-space dimension $27-18+1=10$, so consistency is equivalent to the vanishing of ten
fundamental cycle forms $L_1,\ldots,L_{10}$ in the affine coefficients.

The 1539 lifted CSS equations are denoted by $E_1,\ldots,E_{1539}$.  Exact rational elimination followed by clearing
denominators gives, for each $i$, the integer identity
\begin{equation}
 456 L_i=\sum_{j=1}^{1539}m_{ij}E_j,\qquad m_{ij}\in\mathbb{Z}.
 \label{eq:affine-ceiling-certificate}
\end{equation}
The attached verification program reconstructs the base, the 216 affine variables, all CSS equations, and the ten
cycle forms, and checks these identities without numerical tolerances.  Since the prime divisors of
$456=2^3\cdot3\cdot19$ are at most 19, the factor 456 is invertible in $\mathbb{Z}_P$.  Hence CSS orthogonality
forces every $L_i$ to vanish, and the support in~\eqref{eq:unavoidable-support} lifts to a weight-18 vector
$z\in\ker H_X$.

It remains to show that $z$ is not a stabilizer.  The base $X$-logical with zero-based column support
$\{0,1,67,73,83,304\}$ lies in the kernel of the base $H_Z$.  Selecting all $P$ lift indices above these six
columns therefore gives a vector $x\in\ker H_Z$ for every CPM lift.  Its base support intersects
$\mathcal{S}$ only in column 0, so $xz^T=1$ over $\mathbb{F}_2$.  A vector in the row space of $H_Z$ is orthogonal
to every vector in $\ker H_Z$; consequently $z$ is a nontrivial $Z$-logical operator.
\end{proof}

For the published $P=101$ lift, a complete finite search also supplies the matching lower bound.

\begin{theorem}[Exact distance of the constructed code]
The constructed code has parameters $[[34542,23032,18]]$.
\end{theorem}

\begin{proof}
One explicit realization of the logical operator above has zero-based lifted support
\begin{equation*}
\begin{split}
\{&0,1055,2140,4778,8076,8685,13625,15192,16895,\\
  &17592,18676,19707,22417,24743,26195,30218,31899,34494\}.
\end{split}
\end{equation*}
Direct multiplication gives $H_Xz^T=0$.  The all-sheet lift $x$ used in the preceding proof has weight 606,
satisfies $H_Zx^T=0$, and overlaps $z$ in exactly one coordinate.  This independently verifies that $z$ is a
nontrivial logical operator and gives $d\le18$.

For the lower bound, fix either $H=H_X$ or $H=H_Z$ and let $S$ be a minimum nonzero kernel support.  Every column
meets exactly one check in each of three row groups.  Pairing the variables incident to each occupied check converts
$S$ into a connected, simple, triangle-free cubic graph whose three edge colors are perfect matchings.  Simplicity
and triangle-freeness follow from Tanner girth 8.  This reduction also covers checks incident to four or more
variables, because their even incidence sets are paired before the graph is formed.

The union of two color classes is a disjoint union of alternating even cycles.  Enumerating the third perfect
matching, followed by color-preserving isomorphism reduction, gives the complete graph counts in
Table~\ref{tab:distance-enumeration}.  Each graph was then embedded into each Tanner graph by an exhaustive
backtracking search.  Two verified automorphisms reduce the possible root images to 18 orbits of size 1919,
covering all 34542 columns.  No embedding was found on either side.

\begin{table}[H]
\centering
\caption{Complete enumeration used to exclude nonzero kernel vectors of weight below 18.}
\label{tab:distance-enumeration}
\begin{tabular}{rrrr}
\toprule
Weight & Colored graph types & $\ker H_X$ & $\ker H_Z$\\
\midrule
6  & 1     & none & none\\
8  & 10    & none & none\\
10 & 22    & none & none\\
12 & 226   & none & none\\
14 & 1838  & none & none\\
16 & 25375 & none & none\\
\midrule
Total & 27472 & none & none\\
\bottomrule
\end{tabular}
\end{table}

Every nonzero kernel vector has even weight, and weights 2 and 4 are already excluded by simplicity and
triangle-freeness.  The enumeration therefore proves that both kernels have no nonzero vector below weight 18.
In particular, every nontrivial pure $X$- or $Z$-logical has weight at least 18, which gives $d\ge18$ for the CSS
code.  Combining the two bounds proves $d=18$.  The matrix hashes, logical witness, automorphisms, graph generator,
embedding code, and completed search records are supplied with the source and at the data page in
Ref.~\cite{kasai_homepage}.
\end{proof}

\section{Decoding Method}
\label{sec:decoding-method}

The decoder used in the experiments is a log-domain log-likelihood-ratio (LLR) joint belief-propagation (BP)
decoder followed by deterministic post-processing.
Here joint BP refers to the sum-product decoder on the CSS syndrome-decoding factor graph formed by two Tanner
graphs coupled through the local joint prior at each qubit~\cite{kasai_2026_factor_graph_css_decoding}.
Syndrome failures and logical failures are treated as different outcomes: a trial with a nonzero residual syndrome is
classified as a syndrome failure, while a zero-syndrome residual is tested for logical nontriviality separately.
This separation is necessary because a large residual syndrome is a decoder failure but is not itself evidence for a
low-weight logical operator.

The post-processing stage is designed for the cases where BP leaves a small residual syndrome.
For very small residuals, the decoder first applies deterministic local repair tests that verify the produced
syndrome before applying a correction.
When the residual syndrome has weight two, the decoder also uses an elementary-trapping-set (ETS) library, motivated
by trapping-set analysis for quantum LDPC codes~\cite{raveendran_vasic2021trapping}.
Here an ETS entry is a checked local repair pattern whose check-induced subgraph has degree-one odd checks and degree-two
even checks, and whose odd-check set is exactly the two unsatisfied checks.
If these local and ETS tests do not apply, the decoder can fall back to ordered-statistics-decoding (OSD)-style
repair or to a more general syndrome solve, depending on the experiment being performed.

Each ETS-library entry is represented as $(S,c_0,c_1,V)$, where $S\in\{X,Z\}$ is the side to be repaired,
$c_0,c_1$ are the two odd checks, and $V$ is the variable set to flip.
When the residual syndrome is $\{c_0,c_1\}$, the decoder retrieves $V$ and applies it only after verifying that
$V$ produces exactly the residual syndrome.
Thus a library entry is used as a checked correction rather than as an unconditional flip.

\section{Numerical Results}
\label{sec:verification-measurements}

At $p=0.01$, $10^8$ trials were completed with the LLR joint-BP decoder and deterministic post-processing.
No syndrome failure and no logical failure were observed, so the empirical frame error rate (FER) at this point is
zero.
Using the rough one-sided 95\% zero-failure rule $3/N$, the corresponding experimental upper bound at $p=0.01$ is
about $3\times10^{-8}$.

For higher-FER points, the measurement should not rely on a small number of early failures.
We therefore measured an FER sweep from $p=0.034$ downward.
For the lower-FER points, independent decoder processes were combined only to reduce
wall-clock time; each process used the same LLR joint-BP decoder and deterministic post-processing.
All failures in this sweep were syndrome failures; no logical failure was observed.
The combined point estimates are plotted in Fig.~\ref{fig:local-fer-sweep}.

\begin{figure}[H]
\centering
\includegraphics[width=0.92\linewidth]{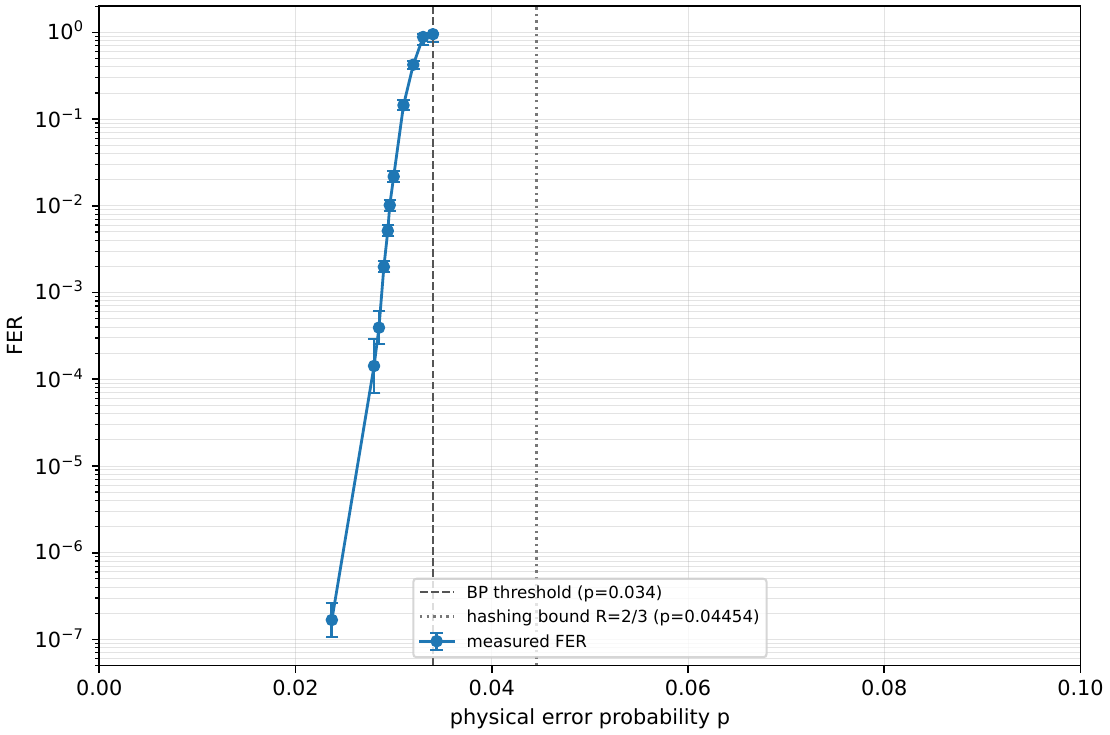}
\caption{FER measurements for the constructed $(3,18)$, $P=101$ lifted code.
The horizontal axis is shown on $0\le p\le0.1$.
The vertical lines mark the BP threshold estimate $p=0.034$ and the rate-$2/3$ hashing bound
$p=0.04454$.  The markers are connected as a visual guide.}
\label{fig:local-fer-sweep}
\end{figure}

The ETS-library branch was tested on failure instances obtained in the $p=0.0237$ measurements.
Five of these failures had $Z$-side residual-syndrome weight two and residual Pauli weight 12.
The corresponding ETS supports all belonged to the same isomorphism orbit.
Using the finite-field translation symmetry of size $q=19$ and the CPM sheet translation of size $P=101$, this orbit
was expanded exhaustively into $19\cdot 101=1919$ ETS entries.
Every expanded entry was validated directly on the lifted $H_X$ Tanner graph by checking that its produced syndrome
is exactly the prescribed two-check syndrome.

For 18 $p=0.0237$ failure instances, the isomorphism-expanded ETS library was combined with local exact repair and weighted-OSD-style
post-processing, 9 of the 18 failures were repaired.
In particular, all five targeted failures with residual-syndrome weight two and residual Pauli weight 12 were repaired.
The remaining nine failures had residual-syndrome weights such as 17, 19, 50, 52, 84, 234, 976, 2367, and 2427, and
are outside the scope of the two-check ETS library used here.

\section{Discussion}

The main feature of the construction is that both stages are expressed as explicit algebraic conditions.
At the base level, regularity, CSS orthogonality, and same-type 4-cycle exclusion reduce to quotient-coset membership
tests.
At the lift level, CSS orthogonality becomes a linear equality system, while same-type 6-cycle closure becomes a family
of nonzero linear-form constraints.
This reduces the finite search space while preserving exact CSS orthogonality and achieving same-type girth 8.

The construction also uses finite-length decoding evidence in a specific way.
Successive lifts are tested by decoding; observed low-weight logical supports are converted into additional forbidden
supports; and the lift search is repeated with the enlarged constraint set.  The exact-distance computation shows
the limitation of this strategy within the affine ansatz: the support in~\eqref{eq:unavoidable-support} is forced by
CSS orthogonality and therefore cannot be removed by changing the affine coefficients.  Avoiding this weight-18 mode
requires a larger shift ansatz, such as non-affine dependence on $t$, followed by new girth and distance checks.

\section{Conclusion}

Using a two-branch multiplicative-coset base over $\mathbb{F}_{19}$ and a $P=101$ CPM lift, we constructed a
$(3,18)$-regular CSS LDPC code with
\begin{equation*}
  n=34542,\qquad k=23032,\qquad k/n=0.6667824677.
\end{equation*}
The code has row weight 18, column weight 3, exact CSS orthogonality, and same-type Tanner girth 8.
The low-weight logical supports accumulated during the iterative construction procedure, together with their affine
equivalents, are not realizable in the constructed lift.
In $10^8$ trials at $p=0.01$, no failure was observed under LLR joint BP with post-processing.
The FER sweep estimates the FER as $5.14\times 10^{-3}$ at $p=0.0294$ and
$3.94\times 10^{-4}$ at $p=0.0285$, with no observed logical failures in the sweep.
A computer-assisted proof establishes the exact parameters $[[34542,23032,18]]$: an explicit weight-18
nontrivial $Z$-logical gives the upper bound, and a complete symmetry-reduced enumeration excludes every nonzero
vector below weight 18 in both parity-check kernels.  More generally, the same base has the structural ceiling
$d\le18$ for every affine-in-$t$ CPM lift satisfying CSS orthogonality whenever the prime lift degree is greater
than 19.

\bibliographystyle{IEEEtran}
\bibliography{refs}

\end{document}